\documentclass[11pt]{article}
\usepackage{amsmath}
\usepackage{amsthm}
\usepackage{amssymb}
\usepackage{graphicx}
\usepackage{times}
\usepackage{verbatim}

\newtheorem{lemma}{Lemma}[section]
\newtheorem{remark}[lemma]{Remark}

\newtheorem{prop}[lemma]{Proposition}
\newtheorem{proposition}[lemma]{Proposition}

\numberwithin{equation}{section}
\newcommand{\E}{{\mathbb E}}
\newcommand{\R}{{\mathbb R}}
\newcommand{\bq}{{\bf q}}
\newcommand{\be}{{\bf e}}
\newcommand{\bp}{{\bf p}}
\newcommand{\bxi}{{\bf \xi}}
\newcommand{\bphi}{{\bf \phi}}

\newcommand{\LL}{{\mathcal L}}

\newcommand{\cR}{{\mathcal R}}
\newcommand{\cH}{{\mathcal H}}

\newcommand{\T}{{\mathbb T}}

\newcommand{\cA}{\mathcal A}
\newcommand{\cS}{\mathcal S}
\newcommand{\cG}{\mathcal G}
\newcommand{\cL}{\mathcal L}
\newcommand{\cN}{\mathcal N}

\newcommand{\cZ}{\mathcal Z}
\newcommand{\p}{{\bf p}}

\begin{document}
\setlength{\baselineskip}{15pt}
\title{ASYMPTOTIC ANALYSIS OF THE  GREEN-KUBO FORMULA}
\author{G.A. Pavliotis \\
        Department of Mathematics\\
    Imperial College London \\
        London SW7 2AZ, UK
                    }
\maketitle

\begin{abstract}
A detailed study of various distinguished limits of the Green-Kubo formula for the self-diffusion coefficient is presented in this paper. First, an alternative representation of the Green-Kubo formula in terms of the solution of a Poisson equation is derived when the microscopic dynamics is Markovian. Then, the techniques developed in~\cite{golden2, AvelMajda91} are used to
obtain a Stieltjes integral representation formula for the symmetric and antisymmetric parts of the diffusion tensor. The effect of irreversible microscopic dynamics on the diffusion coefficient is analyzed and various asymptotic limits of physical interest are studied. Several examples are presented that confirm the findings of our theory. 
\end{abstract}
%
%
%
%
\section{Introduction}
\label{sec:intro}
The two main goals of non-equilibrium statistical mechanics are the derivation of macroscopic equations from microscopic dynamics and the calculation of
transport coefficients~\cite{Balescu1975, balescu97, ResibDeLeen77}. The starting
point is a  kinetic equation that governs the evolution of the distribution
function, such as the Boltzmann, the Vlasov, the Lenard-Balescu or the Fokker-Planck
equation. Although most kinetic equations involve a (quadratically) nonlinear
collision operator,  it is quite often the case that for the calculation of  transport coefficients it is sufficient to consider a linearized collision operator and, consequently, a linearized kinetic equation. In this case,
it is well known that the transport coefficients are related to the eigenvalues
of the linearized collision operator~\cite[Ch. 13]{Balescu1975}, \cite[Ch.
10]{balescu97}. The goal of this article is to present some results on the
analysis of transport coefficients for a particularly simple class of kinetic
equations describing the problem of self-diffusion. 

The distribution function $f(q,p,t)$ of a tagged particle satisfies the kinetic equation
\begin{equation}\label{e:kinetic}
\frac{\partial f}{\partial t} + p \cdot \nabla_q f = Q f,
\end{equation}
where $q, \, p$ are the position and momentum of the tagged particle and
$Q$ is a linear collision operator. $Q$ is a dissipative operator which acts only on the momenta and with only
one collision invariant, corresponding to the conservation the particle density.

The macroscopic equation for this problem is simply the diffusion equation
for the particle density $\rho(p,t) = \int f(p,q,t) \, dp$~\cite{Schroter77}
\begin{equation}\label{e:diffusion}
\frac{\partial \rho}{\partial t} = \sum_{i,j=1}^d D_{ij} \frac{\partial^2 \rho}{\partial
x_i \partial x_j},
\end{equation}
where the components of the diffusion tensor $D$ are the transport coefficients that have to be calculated from the microscopic dynamics.

At least two different techniques have been developed for the calculation of transport coefficients. The first technique is based on the analysis of the kinetic equation~\eqref{e:kinetic} and, in particular, on the expansion of the distribution function in an appropriate orthonormal basis, the basis consisting of the the eigenfunctions of the linear collision operator $Q$
~\cite[Ch. 13]{Balescu1975}, \cite[Ch. 10]{balescu97}. Transport coefficients
are then related to the eigenvalues of the collision operator. The second technique is based on the Green-Kubo formalism~\cite{KuboTodaHashitsume91}. This formalism enables us to express transport coefficients in terms of time integrals of appropriate autocorrelation functions. In particular, the diffusion coefficient is expressed in terms of the time integral of the velocity autocorrelation function
\begin{equation}\label{e:Green_Kubo}
D = \int_0^{+\infty} \langle p(t) \otimes p(0) \rangle \, dt.
\end{equation}
The equivalence between the two approaches for the calculation of transport coefficients, the one based on the analysis of the kinetic equation and the other based on the Green-Kubo formalism has been studied~\cite{Resibois1964}. The Green-Kubo
formalism has been compared with other techniques based on the perturbative
analysis of the kinetic equations, e.g.~\cite{petrosky99b,petrosky99a, brilliantov05}
and the references therein. Many works also exist on the rigorous justification of
the validity of the Green-Kubo formula for the self-diffusion 
coefficient~\cite{JiangZhang03, Durr_al90, Chenal06,Spohn91, Chenal06}. 

Usually the linearized collision operator is taken to be a symmetric operator in some
appropriate Hilbert space. When the collision operator is the (adjoint of
the)  generator of a Markov process (which is the case that we will consider in this paper), the assumption of the symmetry of $Q$ is equivalent to the
reversibility of the microscopic dynamics~\cite{qian02}. There are various cases, however, where the linearized collision operator is not symmetric.
As examples we mention the linearized Vlassov-Landau operator in plasma physics~\cite[Eq.
13.6.2]{Balescu1975} or the motion of a charged particle in a constant magnetic
field undergoing collisions with the surrounding medium~\cite[Seq. 11.3]{balescu97}.
It is one of the main objectives of this paper to study the effect of the
antisymmetric part of the collision operator on the diffusion tensor. 

In most cases (i.e. for most choices of the collision kernel) it is impossible to obtain explicit formulas for transport coefficients. The best one can
hope for is the derivation of estimates on transport coefficients   as functions of the parameters of the microscopic dynamics. The derivation of such estimates
is quite hard when using formulas of the form~\eqref{e:Green_Kubo}. In this
paper we show that the Green-Kubo formalism is equivalent to a formulation
based on the solution of a Poisson equation associated to the collision operator
$Q$. Furthermore,we show that this formalism is a much more convenient starting point
for rigorous and parturbative analysis of the diffusion tensor. The Poisson equation (cell problem) is the standard tool for calculating homogenized coefficients in the theory of homogenization for stochastic differential equations and partial differential equations~\cite{PavlSt08}. 

The problem of obtaining estimates on the diffusion coefficient has been
studied quite extensively in theory of turbulent diffusion--the motion of
a particle in a random, divergence-free velocity field~\cite{kramer}.
In particular, the dependence of the diffusion coefficient (eddy diffusivity)
on the Peclet number has been investigated. For this purpose, a very interesting
theory has been developed by Avellaneda and Majda~\cite{ma:mmert, AvelMajda91},
see also~\cite{golden2, bhatta_1}. This theory is based on the introduction
of an appropriate bounded (and sometimes compact) antisymmetric operator
and it leads to a very systematic and rigorous perturbative analysis of the
eddy diffusivity. This theory has been extended to time-dependent flows~\cite{avellanada2}.

In this paper we apply the Majda-Avellaneda theory to the problem of the
derivation of rigorous estimates for the diffusion tensor of a tagged particle
whose distribution function satisfies a kinetic equation of the form~\eqref{e:kinetic}.
We study this problem when the collision operator is the Fokker-Planck operator
(i.e. the $L^2$-adjoint) of an ergodic Markov process. This assumption is
not very restrictive when studying the problem of self-diffusion of a tagged
particle  since many dissipative integrodifferential operators are generators of Markov processes~\cite{KikuchiNegoro96}.  We obtain formulas for both
the symmetric and the antisymmetric parts of the diffusion tensor and we
use these formulas in order to study various asymptotic limits of physical
interest.

The rest of the paper is organized as follows. In Section~\ref{sec:Green_Kubo}
we obtain an alternative representation for the diffusion tensor based on
the solution of a Poisson equation and we present two elementary examples. In Section~\ref{sec:Diff} we apply the Majda-Avellaneda theory to the problem of self-diffusion and we study rigorously the weak
and strong coupling limits for the diffusion tensor. Examples are presented
in Section~\ref{sec:examples}. Conclusions and open problems are discussed
in Section~\ref{sec:conclusions}.

\section{The Green-Kubo Formula}
\label{sec:Green_Kubo}

In this section we show that we can rewrite the Green-Kubo formula for the diffusion
coefficient in terms of the solution of an appropriate Poisson equation.
We will consider a slight generalization of~\eqref{e:kinetic}, namely we will consider the long-time dynamics of the dynamical system
\begin{equation}\label{e:x}
\frac{d x}{d t} = V(z),
\end{equation} 
where $z$ is an ergodic Markov process state space $\cZ$, generator $\cL$ and invariant measure $\pi(dz)$.\footnote{We remark that the process $z$
can be $x$ itself, or the restriction of $x$ on the unit torus. This is the
precisely the case in turbulent diffusion and in the Langevin equation in
a periodic potential.} The kinetic equation for the distribution function is
\begin{equation}\label{e:kinetic_gen}
\frac{\partial f}{\partial t} + V(z) \cdot \nabla_x f = \cL^* f,
\end{equation}
where $\cL^*$ denotes the $L^2(\cZ)$-adjoint of the generator $\cL$. The
kinetic equation~\eqref{e:kinetic} is of the form~\eqref{e:kinetic_gen} for
$V(p) = p$, and where we assume that the collision operator (which acts only
on the velocities) is the Fokker-Planck operator of an ergodic Markov process,
which can be a diffusion process (e.g. the Ornstein-Uhlenbeck process in
which case~\eqref{e:kinetic_gen} becomes the Fokker-Planck equation) a jump
process (as in the model studied in~\cite{Ellis73}), or a L\'{e}vy process.

\begin{prop}\label{prop:green_kubo}
Let $x(t)$ be the solution of~\eqref{e:x}, let $z(t)$ be an ergodic Markov process with state space $\cZ$, generator $\cL$ and invariant measure $\pi(dz)$ and assume that $V(z)$ is centered with respect to $\mu(dz)$,
$$
\int_{\cZ} V(z) \, \mu(dz) = 0.
$$
 Then the diffusion tensor~\eqref{e:Green_Kubo} is given by
\begin{equation}\label{e:deff}
D = \int_{\cZ} V(z) \otimes \phi(z) \, \mu(dz)
\end{equation}
where $\phi$ is the solution of the Poisson equation
\begin{equation}\label{e:poisson}
- \cL \phi = V(z)
\end{equation}
\end{prop}

\begin{proof}

Let $e$ by an arbitrary unit vector. We will use the notation $D^e = D e \cdot e, \, x^e = x \cdot e$. The Green-Kubo formula for the diffusion coefficient along the direction $e$ is
\begin{eqnarray}
D^e &=& \int_0^{+\infty} \langle \dot{x}^{e} (t) \dot{x}^e (0) \rangle \, dt \nonumber
   \\ & = & \int_0^{+\infty} \langle V^e (z(t)) V^e (z(0))  \rangle \, dt. \label{e:green_kubo1}
\end{eqnarray}
We calculate now the correlation function in~\eqref{e:green_kubo1}. We will use the notation $z=z(t ; p)$ with $z(0;p) = p$. We have
\begin{equation}\label{e:corel}
\langle V^e (z(t;p)) V^e (z(0;p))  \rangle  = \int_{\cZ} \int_{\cZ} V^e (z) V^e (p) \rho(z,t ;p) \mu(d p) dz, 
\end{equation}
where $\rho(z,t ; p)$ is the transition probability density of the Markov process $z$ which is the solution of the Fokker-Planck equation
\begin{equation}\label{e:fokker_planck}
\frac{\partial \rho}{ \partial t} = \cL^* \rho, \quad \rho(z,0;p) = \delta(z-p).
\end{equation}
We introduce the function
\begin{equation*}
\overline{V}^e(t,p) := \E V^e(z) =   \int_{\cZ} V^e (z) \rho(z,t ; p) dz
\end{equation*}
which is the solution of the backward Kolmogorov equation
\begin{equation}\label{e:kolmogorov}
\frac{\partial \overline{V}}{\partial t} = \cL \overline{V}^e, \quad \overline{V}^e (0,p) = V^e(p).
\end{equation}
We can write formally the solution of this equation in the form
$$
\overline{V}^e = e^{\cL t} V^e(p).
$$
We substitute this into~\eqref{e:corel} to obtain
$$
\langle V^e (z(t;p)) V^e (z(0;p))  \rangle = \int_{\cZ} \left( e^{\cL t} V^e(p) \right) V^e (p) \, \mu(dp). 
$$
We use this now in the Green-Kubo formula~\eqref{e:green_kubo1} and, assuming that we can interchange the order of integration, we calculate
\begin{eqnarray*}
D^e & = & \int_0^{+\infty} \left( e^{\cL t} V^e(p) \right) V^e (p) \, \mu(dp) \, dt \\ & = & \int_{\cZ} \left( \int_0^{+\infty} e^{\cL t} V^e(p) \, dt \right) V^e (p) \, \mu(dp) \\ & = & \int_{\cZ} \Big( (-\cL)^{-1} V^e (p) \Big) V^e (p) \mu(d p) \\ & = & \int_{\cZ} \phi^e V^e \, \mu(dp).  
\end{eqnarray*}
where $\phi^e$ is the solution of the Poisson equation $-\cL \phi^e = V^e$. In the above calculation we used the identity $(-\cL)^{-1} \cdot = \int_0^{+\infty} e^{\cL t} \cdot \, dt$~\cite[Ch. 11]{PavlSt08}, ~\cite[Ch. 7]{evans}.

\end{proof}

From~\eqref{e:deff} it immediately follows that the diffusion tensor is nonnegative definite:
\begin{eqnarray*}
D^e:=e \cdot D e &=& \int_{\cZ} V^e \phi \cdot e \, \mu(dz) = \int_{\cZ} (-\cL) \phi^e \phi^e \, \mu(dz) \\ & \geq & 0,
\end{eqnarray*}
since, by definition, the collision operator is dissipative. 

When the generator $\cL$ is a symmetric operator in $L^2(\cZ;\mu(dz))$, i.e.
the Markov process $z$ is reversible~\cite{qian02}, the diffusion tensor is symmetric:
\begin{eqnarray*}
D_{ij} &=& \int_{\cZ} V_i(z) \phi_j(z) \, \mu(dz) = \int_{\cZ} (-\cL) \phi_i(z) \phi_j(z) \, \mu(dz) \\ & = & \int_{\cZ}  \phi_i(z) (-\cL)  \phi_j(z) \, \mu(dz) = D_{ji}.
\end{eqnarray*}
Green-Kubo formulas for reversible Markov processes have already been studied,
since in this case the symmetry of the generator of the Markov process implies
that the spectral theorem for self-adjoint operators can be used~\cite{kipnis,JiangZhang03}. Much less
is known about Green-Kubo formulas for non-reversible Markov process. One of the consequences of non-reversibility, i.e. when the generator of the Markov process $z$ is not symmetric in $L^2(\cZ;\mu(dz))$, is that the diffusion tensor is not symmetric, unless additional symmetries are present. The symmetry properties of the diffusion tensor in anisotropic porous media have been studied in~\cite{KochBrady88}, see also~\cite{thesis}. A general representation formula for the antisymmetric part of the diffusion tensor will be given in the next section.

\subsection{Elementary Examples}

\paragraph{The Ornstein-Uhlenbeck process.}

The equations of motion are
\begin{eqnarray}
\dot{q} & = & p,  \\
\dot{p} & = & -\gamma p + \sqrt{2 \gamma \beta^{-1}} \dot{W}.
\end{eqnarray}
The equilibrium distribution of the velocity process is
$$
\mu(dp) = \sqrt{\frac{\beta}{2 \pi}} e^{-\frac{\beta}{2} p^2} \, dp.
$$

The Poisson equation is
$$
- \cL \phi = p, \quad \cL = - \gamma p \partial_p + \gamma \beta^{-1} \partial_p^2.
$$
The mean zero solution is
$$
\phi = \frac{1}{\gamma} p.
$$
The diffusion coefficient is
$$
D =  \int \phi p \mu (dp) = \frac{1}{\gamma \beta},
$$
which is, of course, Einstein's formula.

\paragraph{A charged particle in a constant magnetic field.}

We consider the motion of a charged particle in the presence of a constant magnetic field in the $z$ direction, ${\bf B} = B e_3$, while the collisions are modeled as white noise~\cite[Ch. 11]{balescu97}. The equations of motion are 
\begin{eqnarray}
\frac{d {\bf q}}{d t} & = & {\bf p}, \\
\frac{d {\bf p}}{d t} & = & \Omega \, {\bf p} \times e_3 - \nu {\bf p} + \sqrt{2 \beta^{-1} \nu} \dot{{\bf W}},
\end{eqnarray}
where ${\bf W}$ denotes standard Brownian motion in $\R^3$, $\nu$ is the collision frequency and 
$$
\Omega = \frac{e B}{m c}
$$
is the Larmor frequency of the test particle.

The velocity is a Markov process with generator
\begin{equation}\label{e:magnetic}
\cL = \Omega (p_2 \partial_{p_1} - p_1 \partial_{p_2}) + \nu (- p \cdot \nabla_p + \beta^{-1} \Delta_p).
\end{equation}
The invariant distribution of the velocity process is the Maxwellian
$$
\mu(d {\bf p}) = \left(\frac{\beta}{2 \pi} \right)^{\frac{3}{2}} e^{-\frac{\beta}{2} | \bf p|^2} \, d {\bf p}.
$$
The vector valued Poisson equation is
$$
- \cL {\bf \phi} = {\bf p}.
$$
The solution is
$$
{\bf \phi} = \left(\frac{\nu}{\nu^2 + \Omega^2} p_1 + \frac{\Omega}{\nu^2 + \Omega^2} p_2, \, -\frac{\Omega}{\nu^2 + \Omega^2} p_1 + \frac{\nu}{\nu^2 + \Omega^2} p_2, \,  \frac{1}{\nu} p_3  \right).
$$
The diffusion tensor is
\begin{eqnarray}
D & = & \int {\bf p} \otimes {\bf \phi} \mu(d {\bf p}) = \beta^{-1} \left(
\begin{array}{ccc}
 \frac{\nu}{\nu^2 + \Omega^2} \; & \;  \frac{\Omega}{\nu^2 + \Omega^2} \; & \; 0 \\
 -  \frac{\Omega}{\nu^2 + \Omega^2} \;  & \;  \frac{\nu}{\nu^2 + \Omega^2} \; & \; 0 \\
  0 \;  & \;  0 \; & \; \frac{1}{\nu} 
\end{array} 
\right) .
\end{eqnarray}
Notice that the diffusion tensor is not symmetric. This is to be expected,
since the generator of the Markov process~\eqref{e:magnetic} is not symmetric.
%
%
%
%
\section{Stieltjes Integral Representation and Bounds on the Diffusion Tensor}
\label{sec:Diff}
When the Markov process $z$ is reversible, it is straightforward to obtain an integral representation formula for the diffusion tensor using the spectral theorem for the self-adjoint operators~\cite{kipnis}. It is not possible, in general, to do the same when $z$ is a nonreversible ergodic Markov process. This problem was solved by Avellaneda and Majda~\cite{AvelMajda91} in the context of the theory of turbulent diffusion by introducing an appropriate bounded, antisymmetric operator. In this section we apply the Avellaneda-Majda theory in order to study the diffusion tensor~\eqref{e:deff} when $z$ is an ergodic Markov process in $\cZ$.

We will use the notation $L^2_{\mu} :=L^2 (\cZ ; \mu(dz))$. We decompose the collision operator $\cL$ into its symmetric and antisymmetric part with respect to the  $L^2_{\mu}$ inner product:
$$
\cL = \cA + \gamma \cS,
$$
where $\cA = -\cA^*$ and $\cS = \cS^*$. The parameter $\gamma$ measures the strength of  the symmetric part, relative to the antisymmetric part. The Poisson equation~\eqref{e:poisson}, along the direction $e$, can be written as
\begin{equation}\label{e:poisson2}
- (\cA + \gamma \cS) \phi^e = V^e.
\end{equation}
Our goal is to study the dependence of the diffusion tensor on $\gamma$, in particular in the physically interesting regime $\gamma \ll 1$.

let $(\cdot, \cdot)_{\mu}$ denote the inner product in $L^2_{\mu}$. We introduce the family of seminorms
$$
\|f \|_k^2 := (f, (-\cS)^k f)_{\mu}.
$$
Define the function spaces $H^k:= \{f \in L^2_\mu \, : \, \|f \|_k < +\infty
\}$
and set $k=1$. Then $\| \cdot \|_1$ satisfies the parallelogram identity and, consequently, the completion of $H^1$ with respect to the norm $\| \cdot \|_1$, which is denoted by $\cH$, is a Hilbert space. The inner product $\langle
\cdot , \cdot \rangle$ in $\cH$   is defined through polarization and it is easy to check that, for $f, \, h \in \cH$,
$$
\langle f,h \rangle = (f, (-\cS) h)_{\mu}.
$$
A careful analysis of the function space $\cH$ and of its dual is presented in~\cite{LanOll05}. 

Motivated by~\cite{AvelMajda91}, see also~\cite{bhatta_1, GoldenPapanicolaou83}, we apply the operator $(-\cS)^{-1}$ to the Poisson equation~\eqref{e:poisson2} to obtain 
\begin{equation}\label{e:poisson3}
 (-\cG + \gamma I) \phi^e = \widehat{V}^e,
\end{equation}
where we have defined the operator $\cG:=(-\cS)^{-1} \cA$ and we have set $\widehat{V}^e := (-\cS)^{-1} V^e$. This operator is antisymmetric in $\cH$:
\begin{lemma}\label{e:antisym}
The operator $\cG: \cH \rightarrow \cH$ is antisymmetric.
\end{lemma}
\begin{proof}
We calculate
\begin{eqnarray*}
\langle \cG f, h \rangle & = & \int (-\cS)^{-1} \cA f (- \cS) h \, \mu (dz) = \int \cA f  h \, \mu(dz) \\ & = & - \int (-\cS)^{-1} (-\cS) f \cA  h \, \mu(dz) = - \int f (-\cS) \cG  h \, \mu(dz) \\ & = & - \langle f, \cG h \rangle.
\end{eqnarray*}
\end{proof}
We remark that, unlike the problem of turbulent diffusion~\cite{AvelMajda91, bhatta_1, MajMcL93}, the operator $\cG$ is not necessarily bounded or, even more, compact. Under the assumption that $\cG$ is bounded as an operator from $\cH$ to $\cH$ we can develop a theory similar to the one developed in~\cite{AvelMajda91}. The boundedness of the operator $\cG$ needs to be checked for each specific example.

Using the definitions of the space $\cH$, the operator $\cG$ and the vector  $\widehat{V}$ we obtain
\begin{equation}\label{e:deff_alt}
D_{ij} = \langle \phi_i, \widehat{V}_j \rangle.
\end{equation}

We will use the notation $\| \cdot \|_{\cH}$ for the norm in $\cH$. It is
straightforward to analyse the overdamped limit $\gamma \rightarrow +\infty$.

\begin{proposition}\label{prop:large_gamma_exp}
Assume that $\cG: \cH \rightarrow \cH$ is a bounded operator. Then, for $\gamma, \, \alpha$ such that $\|\cG \|_{\cH \rightarrow \cH} \leq \gamma$, the diffusion coefficient admits the following asymptotic expansion
\begin{equation}\label{e:deff_exp}
D = \frac{1}{\gamma} \|\widehat{V} \|_{\cH}^2 + \sum_{k=1}^{\infty} \frac{1}{\gamma^{2k+1}} \|\cG^k \widehat{V} \|^2_{\cH}.
\end{equation} 
In particular,
\begin{equation}\label{e:large_gamma}
\lim_{\gamma \rightarrow +\infty} \gamma D = \|\widehat{V} \|_{\cH}^2. 
\end{equation}
\end{proposition}

\begin{proof}
We use~\eqref{e:poisson3}, the definition of the space $\cH$, and the boundendness and antisymmetry of the operator $\cG$ to calculate
\begin{eqnarray*}
D^e & = & \frac{1}{\gamma} \left\langle \left(I - \frac{1}{\gamma} \cG \right)^{-1} \widehat{V}^e, \widehat{V}^e \right\rangle 
\\ & = & 
\frac{1}{\gamma} \sum_{k=0}^{+\infty} \frac{1}{\gamma^k} \left\langle \cG^k \widehat{V}^e , \widehat{V}^e  \right\rangle
\\ & = &
\frac{1}{\gamma} \|\widehat{V}^e \|_{\cH}^2 + \sum_{k=1}^{+\infty} \frac{1}{\gamma^{2 k +1}} \left\langle \cG^{2 k} \widehat{V}^e , \widehat{V}^e  \right\rangle
\\ & = &
\frac{1}{\gamma} \|\widehat{V} \|_{\cH}^2 + \sum_{k=1}^{+\infty} \frac{1}{\gamma^{2 k +1}} \left\| \cG^{ k} \widehat{V}^e  \right\|_{\cH}^2.
\end{eqnarray*}
\end{proof}

From~\eqref{e:large_gamma} we conclude that the large $\gamma$ asymptotics of the diffusion coefficient is universal: the scaling $D^e \sim \frac{1}{\gamma}$ is independent of the specific properties of $\cA, \, \cS$ or $\widehat{V}^e$. This is also the case in problems where the operator $\cG$ is not bounded, such as the Langevin equation in a periodic potential~\cite{HP07}. 

Of course, the expansion~\eqref{e:deff_exp} is of limited applicability, since it has a very small radius of convergence. This expansion cannot be used to study the small $\gamma$ asymptotics of the diffusion coefficient. The analysis of this limit requires the study of a weakly dissipative system, since the antisymmetric part of the generator $\cA$ represents the deterministic part of the dynamics, whereas the symmetric part $\cS$ the noisy, dissipative dynamics. It is well known that the dynamics of such a system in the limit $\gamma \rightarrow 0$ depends crucially on the properties of the unperturbed deterministic system~\cite{freidlin5, FreidWentz84, ConstKiselRyzhZl06}. The properties of this system can be analyzed by studying the operator $\cA$. For the asymptotics of the diffusion coefficient, the null space of this operator has to be characterized. This fact has been recognized in the theory of turbulent diffusion~\cite{AvelMajda91,MajMcL93, kramer}. We will show that a similar theory to the one developed in these papers can be developed in the abstract framework adopted in this paper.

Assume that $\cG: \cH \rightarrow \cH$ is bounded. Let $\cN = \{f \in \cH \, : \, \cG f = 0 \}$ denote the null space of $\cG$. We have $\cH = \cN \oplus \cN^{\bot}$. We take the projections on $\cN$ and $\cN^\bot$ to rewrite~\eqref{e:poisson3} as
\begin{equation}\label{e:poisson_4}
 \gamma \phi_N = \widehat{V}_N, \quad (-\cG +\gamma I) \phi_{N^{\bot}} = \widehat{V}_{N^\bot}.
\end{equation}
We can now write
$$
D^e = \frac{1}{\gamma} \|\widehat{V}_N^e \|_{\cH}^2 + \langle  \phi_{N^\bot}, \widehat{V}^e_{N^{\bot}} \rangle.
$$
\begin{prop}
Assume that there exists a function $p \in \cH$ such that
$$
- \cG p = \widehat{V}^e_{N^\bot}.
$$
Then
\begin{equation}\label{e:gamma_lim}
\lim_{\gamma \rightarrow 0} \gamma D^e = \|\widehat{V}^e_N  \|_{\cH}^2
\end{equation}
In particular, $D^e = o(1/\gamma)$ when $\widehat{V}^e_N = 0$.
\end{prop}

\begin{proof}
We write $\phi_{N^{\bot}} = p + \psi$ where $\psi$ solves the equation
$$
(-\cG + \gamma I) \psi = - \gamma p.
$$
We use $\psi$ as a test function and use the antisymmetry of $\cG$ in $\cH$ to obtain the estimate
$$
\|\psi \|_{\cH} \leq C,
$$
from which we deduce that $\| \phi_{N^{\bot}} \|_{\cH} \leq C$ and~\eqref{e:gamma_lim} follows.
\end{proof}

Let $\cG$ be a bounded operator. Since it is also skew-symmetric, we can write $G = i \Gamma$ where $\Gamma$ is a self-adjoint operator in $\cH$. From the spectral theorem of bounded self-adjoint operators we know that there exists a one parameter family of projection operators $P(\lambda)$ which is right-continuous, and $P(\lambda) \leq P(\mu)$ when $\lambda \leq \mu$ and $P(-\infty)= 0, \, P(+\infty) = I$ so that
$$
f(\Gamma) = \int_{\R} f(\lambda) \, d P(\lambda)
$$
for all bounded continuous functions. Using the spectral resolution of $\Gamma$ we can obtain an integral representation formula for the diffusion coefficient~\cite{AvelMajda91}:
\begin{equation}\label{e:avell_majda}
D^e = \frac{1}{\gamma}  \|\widehat{V}_N^e \|_{\cH}^2 +  2 \gamma \int_0^{+\infty} \frac{ d \mu_e}{\gamma^2 + \lambda^2} 
\end{equation}
where $d\mu_e = \langle  d P(\lambda) \widehat{V}_{N^\bot}^e, \widehat{V}_{N^\bot}^e
\rangle$. We can obtain a similar formula for the antisymmetric part of the diffusion tensor
$$
A = \frac{1}{2} (D - D^T).
$$
In particular, we have the following.
\begin{prop}\label{prop:antisymm}
Assume that the operator $\cG : \cH \rightarrow \cH$ is bounded. Then the
antisymmetric part of the diffusion tensor admits the representation
\begin{equation}\label{e:avell_majda_antisym}
A_{ij} =  \frac{1}{2} \int_{\R} \frac{\lambda d \mu_{ij}(\lambda) }{ \lambda^2
+ \gamma^2}
\end{equation}
where 
$$
d \mu_{ij} = \langle d P (\lambda) \widehat{V}^i_{N^\bot} , \widehat{V}^j_{N^\bot}\rangle.
$$   
\end{prop}

\begin{proof}
We calculate
\begin{eqnarray*}
A_{ij} & = & \frac{1}{2} (D_{ij} - D_{ji})
\\     & = & \frac{1}{2} \big(  \langle \phi_i , \widehat{V}^j \rangle - \langle \phi_j , \widehat{V}^i \rangle  \big)
\\     & = & \frac{1}{2 }  \big(  \langle \phi_i , \widehat{V}^j_{\bot}) -  \langle \phi_j , \widehat{V}^i_{\bot} \rangle  \big)
\\     & = & \frac{1}{2 } \big(   \langle \cR_\gamma \widehat{V}^i_{\bot} , \widehat{V}^j_{\bot}\rangle -  \langle \cR_\gamma \widehat{V}^j_{\bot} , \widehat{V}^i_{\bot}\rangle  \big)
\\     & = & \frac{1}{2 }  \left\langle \big( \cR_\gamma  -\cR_\gamma^*) \widehat{V}^i_{\bot} , \widehat{V}^j_{\bot}
\right\rangle,
\end{eqnarray*}
where we have used the notation $\cR_\gamma = (-i\Gamma + \gamma I)^{-1}$. From
the symmetry of $\Gamma$ we deduce that $\cR^*_\gamma = (i\Gamma +\gamma I)^{-1}$.
Now we use the representation formula
$$
\cR_\gamma = \int_0^{+\infty} e^{-\gamma t} e^{i \Gamma t} dt
$$
to obtain
\begin{eqnarray*}
A_{ij}  & = & \frac{1}{2 } \left\langle  (\cR_\gamma  -\cR_\gamma^*) \widehat{V}^i_{\bot} , \widehat{V}^j_{\bot} \right\rangle \\
        & = & \frac{1}{2 } \left\langle \int_0^{+\infty} e^{-\gamma t} \left(  e^{i\Gamma t} - e^{- i\Gamma t}  \right)   dt \widehat{V}^i_{\bot} , \widehat{V}^j_{\bot} \right\rangle   \\
       & = & \frac{1}{2} \left\langle \int_0^{+\infty} e^{-\gamma t} \sin (\Gamma
       t) \,   dt \widehat{V}^i_{\bot} , \widehat{V}^j_{\bot} \right\rangle 
 \\
       & = & \frac{1}{2} \int_{\R}  \int_0^{+\infty} e^{-\gamma t} \sin(t
       \lambda) \, dt \, d \mu_{ij} (\lambda) \\
       & = & \frac{1}{2} \int_R \frac{\lambda \, d \mu_{ij}(\lambda)} {\gamma^2 + \lambda^2}.
\end{eqnarray*}
\end{proof}

\begin{remark}
The antisymmetric part of the diffusion tensor is independent of the projection of $\widehat{V}$ onto the null space of $\cG$.
\end{remark}

When the operator $\cG : \cH \rightarrow \cH$ is compact we can use the spectral
theorem for the compact, self-adjoint operator $\Gamma = i \cG$ to obtain
an orthonormal basis for the space $\cN^{\bot} $. In this case the integrals
in~\eqref{e:avell_majda} and~\eqref{e:avell_majda_antisym} reduce to sums
and the analysis of the weak noise limit $\gamma \rightarrow 0$ becomes
rather straightforward. The weak noise (large Peclet number)  asymptotics for the symmetric part
of the diffusion tensor for the advection-diffusion problem with periodic
coefficients were studied in~\cite{bhatta_1, MajMcL93}. The asymptotics of
the antisymmetric part of the diffusion tensor for the advection-diffusion problem were studied in~\cite{thesis}.

\section{Examples}
\label{sec:examples}

\paragraph{The generalized Langevin equation.}
The generalized Langevin equation (gLE) in the absence of external forces reads
\begin{equation}\label{e:gLE}
\ddot{q} = - \int_0^t \gamma(t-s) \dot{q}(s) \, ds + F(t),
\end{equation}
where the memory kernel $\gamma(t)$ and noise $F(t)$ (which is a mean zero stationary Gaussian process) are related through the fluctuation-dissipation theorem
\begin{equation}\label{e:fluct_dissip}
\langle F(t) F(s) \rangle = \beta^{-1} \gamma(t-s).
\end{equation}
We approximate the memory kernel by a sum of exponentials~\cite{Kup03},
\begin{equation}
\gamma(t) = \sum_{j=1}^{N} \lambda^2_j e^{-\alpha_j |t|}.
\end{equation}
Under this assumption, the non-Markovian gLE~\eqref{e:gLE} can be rewritten as a Markovian system of equations in an extended state space:
\begin{subequations}
\begin{eqnarray}
\dot{q} &=& p, \label{e:q} \\
\dot{p} &=& \sum_{j=1}^N \lambda_j u_j, \label{e:p} \\  
\dot{u}_j &=& -\alpha_j u_j - \lambda_j p + \sqrt{2 \beta^{-1} \alpha_j } \,\dot{W}_j,\quad j=1,\dots N. \label{e:u}
\end{eqnarray}
\end{subequations}
This example is of the form~\eqref{e:x} with the driving Markov process being $\{p, u_1,\dots u_N \}$. The generator of this process is
\begin{eqnarray*}
\cL  = \Big(\sum_{j=1}^N  \lambda_j u_j \Big)  \frac{\partial}{\partial p} + \sum_{j=1}^N \Big( -\alpha_j u_j \frac{\partial}{\partial u_j} -\lambda_j p \frac{\partial}{\partial u_j} +\beta^{-1} \alpha_j  \frac{\partial^2}{\partial u_j^2} \Big).
\end{eqnarray*}
This is an ergodic Markov process with invariant measure
\begin{equation}\label{e:invmeas_gle}
\rho(p, u) = \frac{1}{\cZ} e^{- \beta \big(\frac{p^2}{2} + \sum_{j=1}^N \frac{u_j^2}{2} \big)},
\end{equation}
where $\cZ = \big(2 \pi \beta^{-1} \big)^{(N+1)/2} $. The symmetric and antisymmetric parts of the generator $\cL$ in $L^2(\R^{N+1}; \rho(p, {\bf u}) dp d {\bf u})$ are, respectively:
$$
\cS =  \sum_{j=1}^N \Big( -\alpha_j u_j \frac{\partial}{\partial u_j} +\beta^{-1} \alpha_j  \frac{\partial^2}{\partial u_j^2} \Big)
$$
and
$$
\cA = \Big(\sum_{j=1}^N  \lambda_j u_j \Big)  \frac{\partial}{\partial p} + \sum_{j=1}^N \Big( -\lambda_j p \frac{\partial}{\partial u_j}  \Big).
$$
It is possible to study the spectral properties of $(-\cS)^{-1} \cA$. However, it is easier to solve the Poisson equation
$$
- \cL \phi =p
$$
and to calculate the diffusion coefficient. The solution of this equation is
$$
\phi = \sum_{k=1}^N \frac{\lambda_k}{\alpha_k} \frac{1}{\sum_{k=1}^N \frac{\lambda^2_k}{\alpha_k} } u_k + p \frac{1}{\sum_{k=1}^N \frac{\lambda^2_k}{\alpha_k} }
$$
The diffusion coefficient is
$$
D = \int_{\R^{N+1}} p \phi \rho(p,u) \, dp du = \beta^{-1}  \frac{1}{\sum_{k=1}^N \frac{\lambda^2_k}{\alpha_k} }.
$$
We remark that, in the limit as $N \rightarrow +\infty$ the diffusion coefficient can become $0$. Indeed,
\begin{eqnarray*}
\lim_{N \rightarrow +\infty} D   = \left\{ \begin{array}{cc}
 \beta^{-1} C \; & \; \sum_{k=1}^{+\infty} \frac{\lambda^2_k}{\alpha_k} = C^{-1}, \\
 0 \; & \; \sum_{k=1}^{+\infty} \frac{\lambda^2_k}{\alpha_k} = +\infty, 
\end{array} 
\right.
\end{eqnarray*}
Thus, phenomena of anomalous diffusion, in particular of subdiffusion, can appear in this simple model. The rigorous analysis of this problem, in the presence of interactions, will be presented elsewhere~\cite{Ottobrepavliotis09}.

\paragraph{The Generalized Ornstein-Uhlenbeck Process.}

We consider the following SDE
\begin{subequations}\label{e:OUgen}
\begin{eqnarray}
\dot{\bq} & = & \bp, \\
\dot{\bp} & = & (\alpha J -\gamma I ) \bp + \sqrt{2 \gamma \beta^{-1}} \, \dot{{\bf W}},
\end{eqnarray}
\end{subequations}
where ${\bq}, \, {\bp} \in \R^d,$ $J=-J^T$,  $ \alpha, \, \gamma >0$ and ${\bf W}$ is a standard Brownian motion on $\R^d$.

The presence of the antisymmetric term $J \bp$ in the equation for $\bp$ implies that the velocity is an {\bf irreversible} Markov process. The generator of the Markov process $\bp$ is
\begin{equation}\label{e;gen_antisymm}
\cL = (\alpha J -\gamma I) p \cdot \nabla_p + \gamma \beta^{-1} \Delta_p.
\end{equation}
It is easy to check that $\nabla_p \cdot(J p e^{-\frac{\beta^{-1}}{2} |p|^2}) = 0.$ Hence, the equilibrium distribution of the velocity process is the same as in the reversible case:
$$
\mu_{\beta}(dp) = \left( \frac{\beta}{2 \pi} \right)^\frac{d}{2} e^{-\frac{\beta}{2} |p|^2} \, dp.
$$
We can decompose the generator $\cL$ into its $L^2(\R^d ; \mu_{\beta}(dp))-$symmetric and antisymmetric parts:
$$
\cL =\alpha \cA + \gamma \cS,
$$
where $\cA = J p \cdot \nabla_p$ and $\cS = -p \cdot \nabla_p +\beta^{-1} \Delta_p.$ 

Using the results from~\cite{Lunardi1997} (or, equivalently, the fact that
the eigenfunctions and eigenvalues of $\cS$ are known) it is possible to show that the operator $\cG = (-\cS)^{-1} \cA$ is bounded from $\cH :=H^1(\R^d ; \mu_{\beta}(dp))$ and the results obtained in Section~\ref{sec:Diff}
apply.~\footnote{Note, however, that this operator
is not compact from $\cH$ to $\cH$.} For this problem we can also  obtain an explicit formula for the diffusion tensor.

\begin{prop}\label{prop:deff}
The diffusion tensor is given by the formula
\begin{equation}\label{e:deff_j}
D = \beta^{-1} (-\alpha J^T + \gamma I)^{-1}.
\end{equation}
\end{prop}

\begin{proof}
The Poisson equation is
$$
-\cL \bphi = \bp,
$$
where the boundary condition is that $\bphi \in (L^2(\R^d ; \mu_{\beta}(dp)))^d$ and we take the vector field $\bphi$ to be mean zero. The solution to this equation is linear in $\bp$:
$$
\bphi = C \bp,
$$ 
for some matrix $C \in \R^{d \times d}$ to be calculated. Substituting this formula in the Poisson equation we obtain (componentwise)
$$
\sum_{k, \ell} Q_{k \ell} p_\ell C_{i k} = p_i,
$$
where the notation $Q = -\alpha J + \gamma I$ was introduced. Notice that, since $Q$ is positive definite, it is invertible. We take now the $(L^2(\R^d ; \mu_{\beta}(dp)))^d$-inner product with $p_m$ (denoted by $\langle \cdot, \cdot \rangle_{\beta}$) and use the fact that $\langle p_{\ell} , p_m \rangle_\beta = \beta^{-1} \delta_{\ell m}$ to deduce
$$
\sum_k Q_{km} C_{ik} =\delta_{im}, \quad i,m=1,\dots d. 
$$
Or,
$$
Q^T C = I,
$$
and, consequently, $C = (Q^T)^{-1} = (-\alpha J^T + \gamma I)^{-1}$. Furthermore,
\begin{eqnarray*}
D_{ij} & = & \left\langle \phi_i , p_j \right\rangle_{\beta} = \sum_k \left\langle C_{ik} p_k , p_j \right\rangle_{\beta} \\ & = & \beta^{-1} \sum_k C_{i k} \delta_{j k} = \beta^{-1} C_{ij},
\end{eqnarray*} 
from which~\eqref{e:deff_j} follows.
\end{proof}

The small $\gamma$-asymptotics of $D$ depends on the properties of the null space of $\cG:=(-S)^{-1}\cA$ or, equivalently, $\cA$. For the problem at hand, it is sufficient to consider the restriction of ($\cN(\cG)$) (or $\cN(\cA)$) onto linear functions in $\bp$. Consequently, in order to calculate $\cN(\cA)$ we need to calculate the null space of $J$, $\cN(J) = \big\{ \bf{b }\in \R^d \; : \; J \bf{b} = 0  \big \}$. 

As an example, consider the case $d=3$ and set
\begin{eqnarray}
J =  \left(
\begin{array}{ccc}
 0 \; & \;  1 \; & \; 1 \\
 -1 \;  & \;  0 \; & \; 1 \\
  -1 \;  & \;  -1 \; & \; 0 
\end{array} 
\right).
\end{eqnarray}
In this case it is straightforward to calculate the diffusion tensor:
$$
D = \frac{1}{\gamma \, \left( 3\,{\alpha}^{2}+{\gamma}^{2} \right)} \left[ \begin {array}{ccc} {\gamma}^{2}+{\alpha}^{2} 
&
-\alpha\, \left( \gamma+\alpha \right) 
 & 
 -\alpha\, \left( \gamma-\alpha \right) 
\\\noalign{\medskip} \alpha\, \left( \gamma-\alpha \right) 
&
\gamma^{2}+\alpha^{2}   
&
- \alpha\, \left( \gamma+\alpha \right) 
\\\noalign{\medskip}
\alpha\, \left( \gamma+\alpha \right) 
& 
\left( \gamma - \alpha \right) 
&
\gamma^{2}+\alpha^{2}
\end {array} \right] 
$$

\begin{figure}
\centerline{
\begin{tabular}{c@{\hspace{2pc}}c}
\includegraphics[width=2.8in, height = 2.8in]{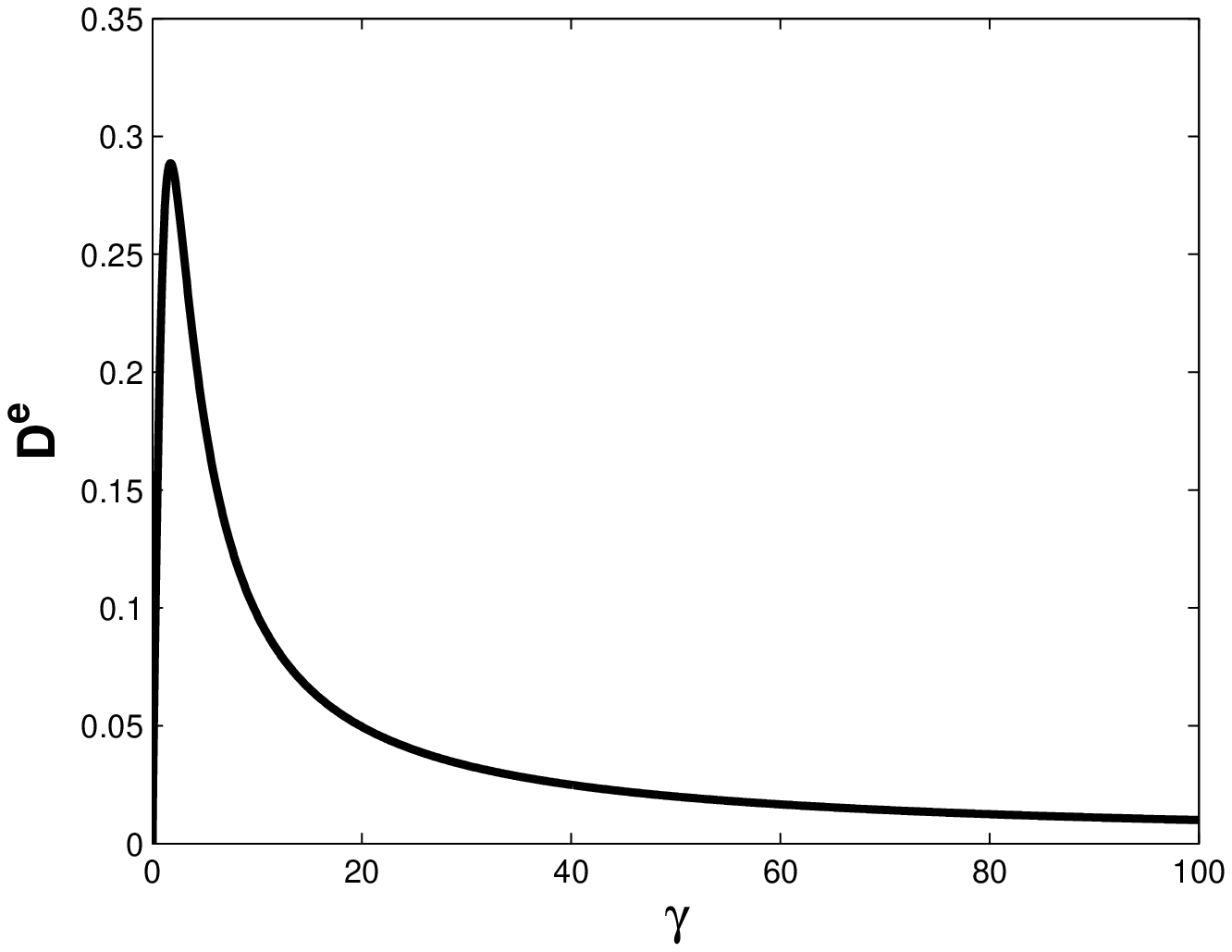} &
\includegraphics[width=2.8in, height = 2.8in]{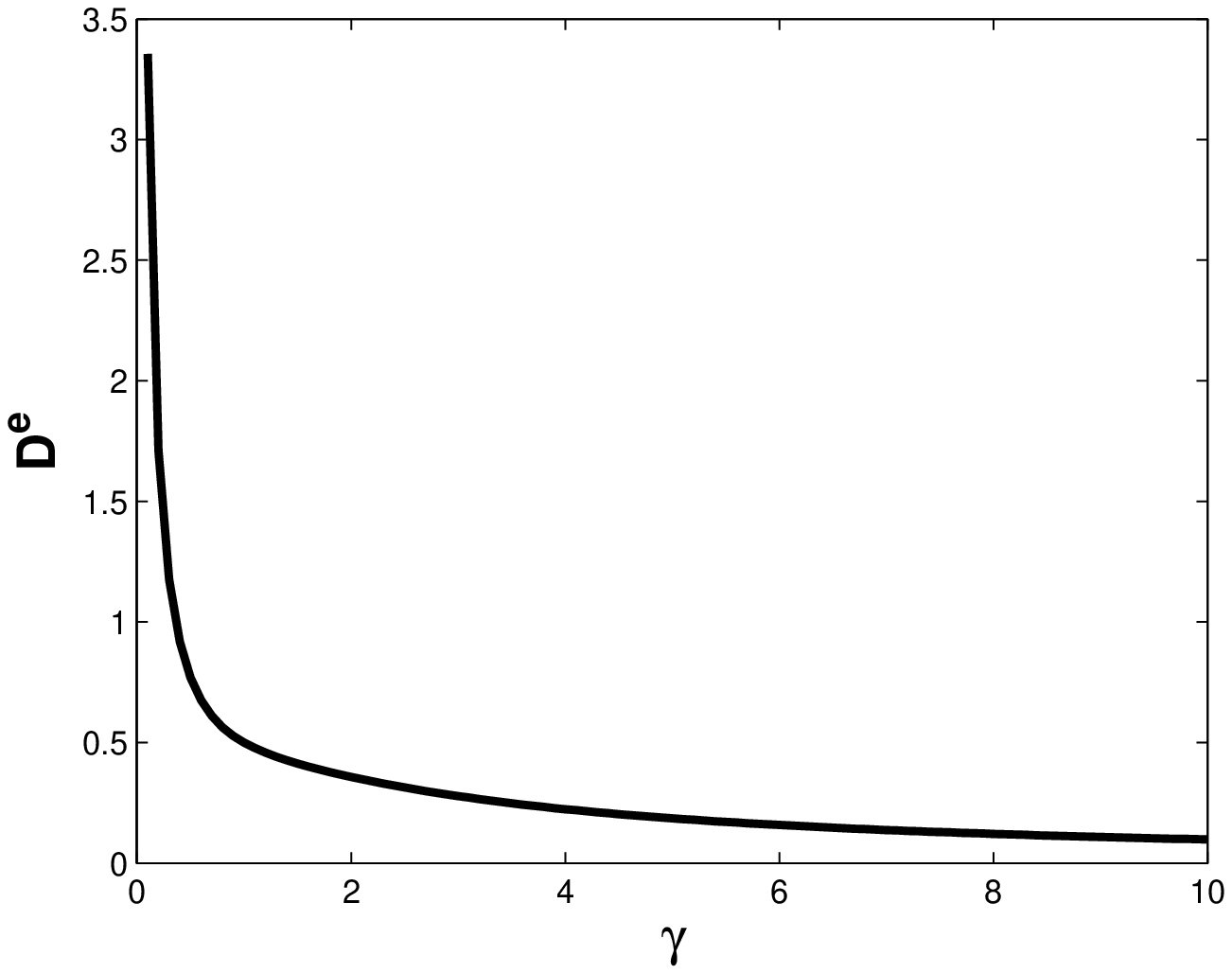} \\
 a.~~  ${\bf e} \cdot \xi = 0$  & b.~~ ${\bf e} \cdot \xi \neq 0$
\end{tabular}}
\begin{center}
\caption{Diffusion coefficient~\eqref{e:deff_e}}
\label{fig:deff}
\end{center}
\end{figure}

The null space of $J$ is one-dimensional and consists of vectors parallel to ${\bf \xi} = (1, \, -1, \, 1).$ From the analysis presented in the previous
section it is expected that the diffusion tensor vanishes in the limit as $\gamma \rightarrow 0$ along directions ${\bf e} \bot {\bf \xi}$. Indeed, from the above formula for the diffusion coefficient we get that (with $|{\bf e}| = 1$)
\begin{equation}\label{e:deff_e}
D^e = \frac{1}{\gamma (3 \alpha^2 +\gamma^2)} \Big( \gamma^2 + |\be \cdot \bxi|^2 \alpha^2 \Big).
\end{equation}
Clearly, when $\be \cdot \bxi = 0$ we have
$$
\lim_{\gamma \rightarrow 0} D^e = 0,
$$
whereas when $\be \cdot \bxi \neq 0$ we obtain
$$
\lim_{\gamma \rightarrow 0} \gamma \, D^e = \frac{|\be \cdot \bxi|^2}{3}.
$$
The diffusion coefficient, as a function of $\gamma$, and for $\alpha =1$
is plotted in Figure~\ref{fig:deff}.
%
%
%
%
\section{Conclusions}\label{sec:conclusions}

The Green-Kubo formula for the self-diffusion coefficient was studied in
this paper. It was shown that the Green-Kubo formula can be rewritten in
terms of the solution of an Poisson equation when
the collision operator is linear and it is the generator of an ergodic Markov
process. Furthermore, the effect of irreversibility in the microscopic dynamics
on the diffusion coefficient was investigated and the Majda-Avellaneda theory
was used in order to study various asymptotic limits of the diffusion tensor.
Several examples were also presented.

There are several directions in which the work reported in this paper can
be extended. First, a similar analysis can be applied to the linear Boltzmann
equation (i.e. for a collision operator that has five collision invariants), in order to obtain alternative representation formulas for other
transport coefficients, in addition to the self-diffusion coefficient. In
this way, it should be possible to obtain rigorous estimates on other transport
coefficients. 

Second, the effect of external forces on the scaling of transport
coefficients with respect to the various parameters of the problem can be
studied: the techniques presented in this paper are applicable to a kinetic
equation of the form
\begin{equation}\label{e:kinetic_V}
\frac{\partial f}{\partial t} + F(q) \cdot \nabla_p f   +   \p \cdot \nabla_q f = Q f,
\end{equation}
where $F(q)$ is an external force.

Finally, phenomena of subdiffusion (i.e. the limit $D \rightarrow 0$) and
superdiffusion (i.e. $D \rightarrow \infty$)  can also be analyzed within the framework developed in this paper. A simple example was given in Section~\ref{sec:examples}. All these problems are currently under
investigation.

\bigskip

\paragraph{Acknowledgments.} The author thanks P.R. Kramer for many useful discussions and comments and for an extremely careful reading of an earlier
version of this paper.

\def\cprime{$'$} \def\cprime{$'$} \def\cprime{$'$} \def\cprime{$'$}
  \def\cprime{$'$} \def\cprime{$'$} \def\cprime{$'$}
  \def\Rom#1{\uppercase\expandafter{\romannumeral #1}}\def\u#1{{\accent"15
  #1}}\def\Rom#1{\uppercase\expandafter{\romannumeral #1}}\def\u#1{{\accent"15
  #1}}\def\cprime{$'$} \def\cprime{$'$} \def\cprime{$'$} \def\cprime{$'$}
  \def\cprime{$'$} \def\cprime{$'$} \def\cprime{$'$}
  \def\polhk#1{\setbox0=\hbox{#1}{\ooalign{\hidewidth
  \lower1.5ex\hbox{`}\hidewidth\crcr\unhbox0}}}

\end{document}